\documentclass[english,prl]{revtex4}
\pdfoutput=1
\usepackage[T1]{fontenc}
\usepackage[utf8]{inputenc}
\usepackage{color}
\usepackage{babel}
\usepackage{amsthm}
\usepackage{amsmath}
\usepackage{amssymb}
\usepackage[unicode=true,
 bookmarks=true,bookmarksnumbered=false,bookmarksopen=false,
 breaklinks=false,pdfborder={0 0 1},backref=false,colorlinks=true]
 {hyperref}
\hypersetup{pdftitle={Sealed States And Quantum Blackmail},
 pdfauthor={Andrew Lutomirski}}

\makeatletter
\@ifundefined{textcolor}{}
{%
 \definecolor{BLACK}{gray}{0}
 \definecolor{WHITE}{gray}{1}
 \definecolor{RED}{rgb}{1,0,0}
 \definecolor{GREEN}{rgb}{0,1,0}
 \definecolor{BLUE}{rgb}{0,0,1}
 \definecolor{CYAN}{cmyk}{1,0,0,0}
 \definecolor{MAGENTA}{cmyk}{0,1,0,0}
 \definecolor{YELLOW}{cmyk}{0,0,1,0}
}
  \theoremstyle{plain}
  \newtheorem*{thm*}{\protect\theoremname}
\theoremstyle{plain}
\newtheorem{thm}{\protect\theoremname}
  \theoremstyle{definition}
  \newtheorem{defn}[thm]{\protect\definitionname}

\usepackage{microtype}

\usepackage{mleftright}
\mleftright 

\makeatother

  \providecommand{\definitionname}{Definition}
  \providecommand{\theoremname}{Theorem}
\providecommand{\theoremname}{Theorem}

\begin{document}

\title{Sealed States And Quantum Blackmail}

\author{Andrew Lutomirski}

\affiliation{AMA Capital Management LLC, Palo Alto, CA}

\email{andy@luto.us}

\selectlanguage{english}%

\date{August 27, 2013}
\begin{abstract}
Consider a protocol in which Belinda seals a (classical) message.
She gives the resulting sealed message to Charlie, who can either
unseal and read the message or return it unopened to Belinda. If he
returns it unopened, Belinda should be able to verify that Charlie
neither read the message nor made a copy that would allow him to read
it later. Such a protocol is impossible with classical cryptography:
Charlie can copy a message and do anything he likes to that copy without
damaging the original. With quantum cryptography, on the other hand,
the no cloning theorem implies that Charlie cannot simply copy a message
and unseal the copy.

In this paper, I prove that any conventional quantum cryptographic
protocol can give at best a very weak security guarantee. However,
quantum cryptography in conjunction with classical functions that
can only be inverted by humans (i.e. CAPTCHAs) can potentially give
exponential security.
\end{abstract}
\maketitle

\section{Introduction}

\global\long\def\ii{\mathbb{I}}
\global\long\def\op#1{\mathop\textrm{#1}}

Imagine that you want to blackmail someone with incriminating documents.
You are worried that your victim might kill you instead of paying
up; for insurance, you give a copy of your incriminating documents
to your attorney to be unsealed and published in the event of your
untimely death. You meet with your victim and extort some quantum
money.

Of course, your reputation as an honest blackmailer would be destroyed
if the incriminating documents got out even though your victim paid
up. To be safe, you ask your attorney for the documents back.

Is there any way your attorney can prove that he or she did not open
and copy the documents before returning them?

With only classical techniques, the best you can do is to physically
seal the documents before giving them to your attorney and check that
the seal is intact when you get the documents back. If your attorney
is good at covertly opening envelopes or if you want to send the documents
to your attorney over the internet instead of in person, this is not
good enough. Encrypting the documents would not help because your
attorney needs to be able to read them if you die.

By the no-cloning theorem, it is conceivable that you could encode
the documents into a \emph{sealed state} and give a subsystem of that
state to your attorney. He or she cannot directly copy the state,
so perhaps any attempt to read the state would be detectable.

A quantum blackmail protocol has two players: Belinda the blackmailer
and Charlie the co-conspiring attorney. Belinda has some secret message
and, from that message, produces a state on two registers, $B$ and
$C$. She gives register $C$ to Charlie. Charlie can do one of three
things:
\begin{itemize}
\item He can \emph{unseal} the state to learn the message.
\item He can give the state in register $C$ back to Belinda and tell her
that he did not unseal the message.
\item He can cheat, trying to learn something about the message while lying
to Belinda and saying he hasn't.
\end{itemize}
If Charlie unseals the state, then the protocol is done; Belinda doesn't
need to do anything. If Charlie tells Belinda that he did not unseal
the message, then Belinda will make a measurement to decide whether
to believe him. If, in fact, Charlie did not cheat, then Belinda will
believe him with probability $1-\epsilon_{c}$, where $\epsilon_{c}$
is the protocol's completeness error. If Charlie did cheat, then Belinda
may or may not catch him and Charlie may learn something about the
message.

In the straightforward case, the message being sealed is a single
classical message, and Charlie must be able to decode the message
correctly with some probability $p$ if he chooses to unseal the state.
If he cheats, then Belinda should catch him with high probability.
This turns out to be impossible to achieve with good security: in
quantum computing, any measurement that has a nearly deterministic
outcome can be performed with almost no damage to the state. If $p$
is large, Charlie can learn the message with minimal damage to his
state and is unlikely to be caught. If $p$ is small, the protocol
is not very useful, and Charlie's chance of getting away with cheating
does not decrease rapidly as $p$ decreases. In particular, Charlie
can always recover the message with probability $p$, and Belinda
will catch him with at most probability $\epsilon_{c}+p\sqrt{1-p}+\left(1-p\right)$.
Charlie's cheating strategy does not require difficult computation.
The proof is given in Section~\ref{sec:blackmail-soundness-bound}.

Nevertheless, there are at least two ways to achieve secure quantum
blackmail protocols.

The simpler case is when Belinda has a collection of distinct pieces
of classical information and unsealing the state only needs to reveal
one of them to Charlie. This is realistic: if Belinda has several
compromising pictures of her victim, then it may be sufficient for
Charlie to have access to a single random picture as insurance. Belinda
can put a random picture in the $C$ register and keep the purification
of the state in the $B$ register; if Charlie copies the picture,
he will know which picture it was and break the entanglement between
the registers. Section~\ref{sec:blackmail-multiple-pictures} discusses
the security properties of this protocol.

The more complicated case uses a new type of quantum computing resource:
a computation that cannot be done on a quantum computer. The generic
attack against sealed states relies on \emph{reversibly} unsealing
the state. If part of the unsealing process cannot be done on a quantum
computer, then the attack may fail. Of course, any calculation that
can be done on a classical computer can, in principle, be done on
a quantum computer as well. Until someone invents perfect artificial
intelligence, however, there will be calculations that can be done
by humans but not by any computer. Since humans are not coherent quantum
computers, we can use such a calculation to implement general-purpose
sealed states. The idea is to encode a message as a superposition
of many different images or sounds, any of which can be decoded by
a human but not by a computer. Since humans cannot be run in reverse,
showing the image or sound to a human will break the superposition.
Section~\ref{sec:blackmail-captchas} gives an example of this type
of protocol.

\section{A bound on soundness and completeness}

\label{sec:blackmail-soundness-bound}Suppose that Belinda has a single
classical message $m$. She prepares a sealed message in register
$C$ and gives it to Charlie. Without loss of generality, we assume
that Belinda keeps a purification of Charlie's state in register $B$,
so the full state is $|\psi_{m}\rangle_{BC}$. Charlie knows some
efficient algorithm that takes register $C$ as input and outputs
$m$ with probability $p$.

If Charlie returns register $C$ unmodified to Belinda, then Belinda
will perform some measurement on registers $B$ and $C$ to determine
whether Charlie cheated. If Charlie did not cheat, then Belinda will
believe him with probability $1-\epsilon_{c}$, where $\epsilon_{c}$
is the completeness error of the protocol.

Charlie can try to cheat. For example, he could perform some measurement
on register $C$ to try to learn $m$ and give the state that remains
in register $C$ back to Belinda. If he tells her that he did not
look at the message, she will believe him with probability $1-s$,
where $s$ is the soundness of the protocol. In general, $s$ can
depend on Charlie's cheating strategy.
\begin{thm*}
(No quantum blackmail) Charlie has an efficient strategy that recovers
$m$ with probability $p$ such that Belinda will catch him with probability
$s\le\epsilon_{c}+p\sqrt{1-p}+\left(1-p\right)$. \end{thm*}
\begin{proof}
If Charlie were to open the message instead of cheating, he would
perform some measurement. Without loss of generality, that measurement
consists of a unitary operator $U$ followed by a projective measurement
$\left\{ P_{i}\right\} $, both on register $C$ %
\footnote{If unsealing the state requires ancilla qubits, then those qubits
can be treated as part of register $C$. The final state of those
qubits does not matter, since Belinda does not expect Charlie to preserve
anything after unsealing the state.%
}. The projective measurement has one outcome per possible message.

Charlie can cheat with a simple algorithm. Charlie acts only on register
$C$, but we will keep track of the joint state of both registers.
The initial state is $|\psi\rangle_{BC}$. Charlie applies $\ii\otimes U$
and then makes the projective measurement $\left\{ \ii\otimes P_{i}\right\} $.
With probability $p$, Charlie gets the outcome corresponding to $m$
and learns the state. With probability $1-p$, Charlie gets a different
outcome and does not learn the state. In either case, Charlie applies
$\ii\otimes U^{\dagger}$ to the state that remains after the measurement
and gives that state to Belinda.

Let $|\phi_{i}\rangle$ be the final state of registers $B$ and $C$
conditioned on Charlie obtaining outcome $i$. Suppose that outcome
$i$ occurs with probability $q_{i}$, so that $p=q_{m}$. Marginalizing
over the outcomes of Charlie's measurement, Belinda receives the mixed
state $\sigma=\sum_{i}q_{i}|\phi_{i}\rangle\langle\phi_{i}|$.

The trace distance between the initial (i.e. untampered-with) state
and the state that Belinda tests is 
\begin{align*}
 & D\left(|\psi\rangle\langle\psi|,\sigma\right)\\
\le\, & \sum_{i}q_{i}D\left(|\psi\rangle\langle\psi|,|\phi_{i}\rangle\langle\phi_{i}|\right)\textrm{ by convexity of the trace distance}\\
\le\, & q_{m}D\left(|\psi\rangle\langle\psi|,|\phi_{m}\rangle\langle\phi_{m}|\right)+\sum_{i\ne q_{m}}q_{i}\\
=\, & pD\left(|\psi\rangle\langle\psi|,|\phi_{m}\rangle\langle\phi_{m}|\right)+\left(1-p\right)\\
=\, & p\sqrt{1-\left|\langle\psi|\phi_{m}\rangle\right|^{2}}+\left(1-p\right)\\
=\, & p\sqrt{1-p}+\left(1-p\right)
\end{align*}

The trace distance is an upper bound on the difference between the
probability that Belinda accepts $|\psi\rangle$ and the probability
that Belinda accepts $|\phi\rangle$. So $s-\epsilon_{c}\le p\sqrt{1-p}+\left(1-p\right)$.
\end{proof}

This bound is rather weak. The derivation assumes that Belinda has
access to Charlie's entire final state, including any ancilla qubits;
in practice, if Charlie uses ancillas, he would likely be better off
zeroing them before returning them to Belinda.

It is also not obvious that even this bound is achievable in practice.
A straightforward protocol using no ancillas comes up far short. Suppose
that Belinda sends Charlie the $C$ register of 
\[
|\psi_{m}\rangle_{BC}=\frac{1}{\sqrt{2}}\left[|0\rangle_{B}|0\rangle_{C}+|m\rangle_{B}|m\rangle_{C}\right],
\]
 where $|0\rangle$ is an arbitrary nonsense state. If Charlie returns
the state unopened, Belinda verifies it by projecting onto $|\psi_{m}\rangle$.
This protocol is complete (i.e.\ $\epsilon_{c}=0$). If Charlie measures
his qubit in the computational basis, Belinda will detect his cheating
with probability 50\%, which is worse than the bound of 85\%.

An improved protocol would replace $|0\rangle_{B}|0\rangle_{C}$ with
a superposition of many garbage states. This would allow Belinda to
detect Charlie's cheating with probability near one if Charlie gets
unlucky and fails to recover the message; that is, the second approximation
above would be tight. This would improve the probability of detection
to 75\%.

Regardless of how close to the soundness bound a protocol of this
type is, it will be subject to other types of attack. For example,
if Charlie is worried that Belinda is blackmailing his friend, then
he can verify that $m$ is \emph{not} a compromising message about
his friend. That is, Charlie can define a classical function $g(m')\in\left\{ 0,1\right\} $
over the space of possible messages such that $g(m')=1$ indicates
that $m'$ is a picture of Charlie's friend. Then Charlie can apply
the message recovery protocol coherently and measure $g$ on the result.
If the result is 1, then Belinda is probably blackmailing Charlie's
friend, and Belinda has a reasonable change of detecting that Charlie
cheated. If, on the other hand, the result was zero, then Charlie
did not damage the state at all and will not be detected.

To work around these attacks, we need to look at a broader class of
protocols.

\section{Multiple compromising pictures}

\label{sec:blackmail-multiple-pictures}If Belinda has more than one
message and she is willing to give Charlie a sealed state that can
only reveal one of them, then she can escape the no-quantum-blackmail
bound. For example, suppose she has $n$ compromising pictures $m_{1},\ldots,m_{n}$
of her victim, and she considers the threat of any one of them being
published to be adequate to ensure her safety. Belinda can generate
a superposition of pictures 
\[
|\psi\rangle_{BC}=\frac{1}{\sqrt{n}}\sum_{i=1}^{n}|i\rangle|m_{i}\rangle
\]
 and send register $C$ to Charlie. Of course, any tabloid would take
great pains to authenticate a damaging photograph before publishing
it, so we will assume that Charlie has no interest in cheating unless
he can recover an entire intact photograph. If he can do this, then
he knows which photograph Belinda sent him, and register $B$ must
collapse to a single value of $i$. No matter what state Charlie sends
Belinda, Belinda can detect that he cheated with probability at least
$\frac{n-1}{n}$. As $n$ becomes large, this protocol becomes more
secure.

This protocol is vulnerable to other attacks, though. Charlie could
measure a property that all $n$ pictures have in common without damaging
the state at all. For example, he could use a quantum image recognition
program to figure out who is in the pictures.

To prevent this type of attack and allow exponential security, we
need a new type of resource.

\section{CAPTCHAs}

\label{sec:blackmail-captchas}The fundamental weakness of quantum
cryptography that makes quantum blackmail difficult is that any computation
can be performed coherently. If Charlie can calculate something based
on a sealed state that Belinda gives him, then he can do the same
calculation coherently, measure some property of the result, and undo
the original calculation. If the unsealing process required that Charlie
does something that cannot be done on a quantum computer, then we
could block these attacks. If quantum computers ever become as powerful
as classical computers, then the only things that quantum computers
will not be able to do are things that no computer at all can do.

CAPTCHAs, a type of spam-preventing technology on the internet, are
based on a computation that can only be reliably done by humans. A
CAPTCHA is a ``Completely Automated Public Turing test to tell Computers
and Humans Apart'' \cite{vonahn2003captcha}. To use a CAPTCHA, a
website generates a random word or number $x$. The website then computes
a picture, sound, or other message based on $x$ that is meant to
be decoded by a human. Any human should be able to find $x$ by inspecting
the web page, but no polynomial-time algorithm should be able to find
$x$ with non-negligible probability.

A CAPTCHA is defined as a communication protocol, possibly with multiple
rounds. We need a function instead of a communication protocol, and
we need the function to be secure against quantum adversaries. We
therefore propose the following definition:
\begin{defn}
A human-invertible one-way function is a one-to-one function $f:\{0,1\}^{*}\to I$
that maps strings of bits to a set $I$ and has these properties:\end{defn}
\begin{itemize}
\item The function $f$ can be evaluated in quantum polynomial time.
\item Given $f(x)$ for unknown $x$, a human acts as an oracle that computes
$x$ in one query with negligible probability of error. This process
is incoherent: $f(x)$ leaks to the environment. Most likely, $I$
will be a set of images, sounds, or other media that a human can look
at.
\item There is no polynomial-time quantum algorithm that can compute $f^{-1}(z)$
with non-negligible probability of success. This must remain true
even if the algorithm can query $f^{-1}$ on other inputs in the range
of $f$.
\end{itemize}

Standard CAPTCHAs are randomized, but we will assume that $f$ can
derive any randomness it needs from its argument.

A human-invertible one-way function is very similar to a trapdoor
one-way function; the difference is that instead of being invertible
using a secret, it is invertible by asking a human to invert it. Without
access to a human, it works just like a trapdoor one-way function
with an unknown secret. We will therefore construct a cipher from
$f$ in much the same way that public-key ciphers are constructed
using trapdoor functions. The standard construction is Bellare and
Rogaway's OAEP \cite{bellare1995oaep}, and we will use a similar
construction. OAEP is randomized, and we will replace the randomness
with entanglement. To unseal a message, Charlie must show a particular
value of $f$ to a human, breaking the entanglement in the process.

Let $n$ be the length of the message being sealed. Following OAEP,
choose security parameters $k$ and $k_{0}$, where $n=k-k_{0}$.
Let $G:\left\{ 0,1\right\} ^{k_{0}}\to\left\{ 0,1\right\} ^{n}$ be
a pseudorandom generator and $H:\left\{ 0,1\right\} ^{n}\to\left\{ 0,1\right\} ^{k_{0}}$
be an ideal hash function. Both $G$ and $H$ are publicly known.
The sealed state encoding $y$ is 
\[
|\psi_{y}\rangle_{BC}=\frac{1}{\sqrt{2^{k_{0}}}}\sum_{r\in\left\{ 0,1\right\} ^{k_{0}}}|r\rangle_{B}|\mathcal{E}_{r}^{G,H}(y)\rangle_{C},
\]
 where 
\[
\mathcal{E}_{r}^{G,H}(y)=f(y\oplus G(r)\parallel r\oplus H(y\oplus G(r)))
\]
 is the encrypted version of $y$. The $\parallel$ operator denotes
concatenation of strings of bits, and $\oplus$ is bitwise exclusive
or.

To seal a message $y$, Belinda generates $|\psi_{y}\rangle_{BC}$
and gives Charlie the $C$ register. To unseal the message, Charlie
measures register $C$ and shows it to a human. The human inverts
$f$ to recover $y\oplus G(r)\parallel r\oplus H(y\oplus G(r))$ for
some unknown value of $r$. Charlie then computes $H(y\oplus G(r))$
to recover $r$, $G(r)$, and $y$.

If Charlie does not unseal the message, he returns register $C$ to
Belinda, who measures $|\psi_{y}\rangle\langle\psi_{y}|$ on the combined
state in registers $B$ and $C$. If the outcome is $1$, then she
believes Charlie; if not, she accuses Charlie of cheating. This protocol
is fully complete -- if Charlie does not cheat, then Belinda will
believe him with probability one.

If Charlie cheats, he can use an arbitrarily complicated strategy
to select the values of $f$ that he asks a human to invert. Let $Q\subseteq I$
be the set of all $f$ values that he will ever ask a human to invert.
Of course, Charlie can ask a human to evaluate $f^{-1}$ even after
Belinda has made her measurement, so Belinda cannot possibly know
the set $Q$. Nonetheless, Charlie is constrained to make a polynomial
number of queries, so $\left|Q\right|=O(\op{poly}(k,k_{0}))$. Let
\[
R=\left\{ r:\mathcal{E}_{r}^{G,H}(y)\in Q\right\} 
\]
 and define projectors 
\begin{align*}
T & =|\psi_{y}\rangle\langle\psi_{y}|=\frac{1}{2^{k_{0}}}\sum_{r,r'\in\{0,1\}^{k_{0}}}|r\rangle_{B}|\mathcal{E}_{r}^{G,H}(y)\rangle_{CC}\langle\mathcal{E}_{r'}^{G,H}(y)|_{B}\langle r'|\text{ and}\\
U & =\frac{1}{2^{k_{0}}-\left|R\right|}\sum_{r,r'\in\{0,1\}^{k_{0}}\setminus R}|r\rangle_{B}|\mathcal{E}_{r}^{G,H}(y)\rangle_{CC}\langle\mathcal{E}_{r'}^{G,H}(y)|_{B}\langle r'|.
\end{align*}

Belinda's measurement is $T$, the projector onto the uniform superposition
of all $r$ values, each paired with its corresponding $\mathcal{E}_{r}$
value. The projector $U$ is almost the same; it projects onto the
uniform superposition of $r$ values that are useless to Charlie --
this is the set of values of $r$ for which Charlie will never ask
a human to invert $f(\mathcal{E}_{r}^{G,H})$. If Belinda somehow
measured $U$ and obtained the outcome 1, then Charlie could only
attempt to decrypt $f(\mathcal{E}_{r}^{G,H})$ without asking a human
to evaluate $f^{-1}(f(\mathcal{E}_{r}^{G,H}))$.

Belinda cannot measure $U$ because she does not know the set $R$,
but she can measure $T$ instead. $T$ and $U$ are both rank 1 projectors,
and they project onto states that differ negligibly from each other.
If Belinda obtains the outcome 1 from $T$ and therefore believes
Charlie, then either she got unlucky due to the difference between
$T$ and $U$, which occurs with negligible probability, or Charlie's
queries are all useless.~%
\footnote{This is not the same thing as saying that if Belinda believes Charlie,
then Charlie's queries are useless with high probability. Charlie
could, for example, unseal the message and give Belinda the classical
state $|\mathcal{E}_{r}^{G,H}(y)\rangle$ in register $C$. Belinda
will believe him with negligible probability, but if she believes
him then his queries are still useful with probability~1.%
} 

In the event that Charlie's queries are useless, then recovering $x$
is mostly equivalent to breaking OAEP, which should be impossible
as long as the human-invertible one-way function $f$ is secure. The
standard OAEP security proof~\cite{bellare1995oaep} fails in this
context for two reasons: it assumes a classical adversary, and it
assumes that the trapdoor function is a permutation; $f$ is merely
one-way. The latter problem should be easily correctable, but the
former will be more challenging. I am unaware of any meaningful security
proofs of OAEP (or, for that matter, any other public-key cipher construction)
against quantum adversaries.

This protocol is also somewhat resistant to attacks that break the
human-invertible one-way function after Belinda makes her measurement
-- if Belinda's test passes with non-negligible probability, then
Charlie does not know anything that would let him reliably determine
the value of $r$. It is hard to imagine that any quantum state he
generated from register $C$ without being able to invert $f$ would
let him later recover $y$ even with unlimited computational power.

If Belinda wants to seal a very long message, it could be more efficient
to generate a short random key, encrypt the message against the key,
and seal the key instead of the message.

\section{Open questions}

These protocols require Belinda to store a quantum state that is entangled
with the sealed state. It should be possible to eliminate the entanglement.
One naive approach is to eliminate register $B$, making the sealed
state 
\[
\frac{1}{\sqrt{2^{k_{0}}}}\sum_{r\in\left\{ 0,1\right\} ^{k_{0}}}|\mathcal{E}_{r}^{G,H}(y)\rangle.
\]
 This is neither practical nor secure: preparing this state is likely
as hard as index erasure \cite{amrr2011indexerasure}, and if Charlie
unseals the message then he can recreate the original state with whatever
algorithm Belinda used to prepare it in the first place. The difficulty
in preparing the state can be resolved by giving Charlie both registers,
that is 
\[
\frac{1}{\sqrt{2^{k_{0}}}}\sum_{r\in\left\{ 0,1\right\} ^{k_{0}}}|r\rangle|\mathcal{E}_{r}^{G,H}(y)\rangle.
\]
 This is insecure for the same reason. An improved version would be
\[
\frac{1}{\sqrt{2^{k_{0}}}}\sum_{r\in\left\{ 0,1\right\} ^{k_{0}}}e^{i\phi(r)}|r\rangle|\mathcal{E}_{r}^{G,H}(y)\rangle,
\]
 where the function $\phi(r)$ is a secret known only to Belinda.
This seems likely to be secure, although the security proof would
be more complicated.

Regardless of whether Belinda needs to store a quantum state, all
of these protocols involve giving Charlie a highly entangled state.
A protocol in which Charlie's state was a tensor product of a large
number of low-dimension systems would be very interesting, since it
could be built with much simpler quantum computing technology. Designing
such a protocol that is secure even if Charlie can make entangling
measurements may be difficult.

Putting aside quantum blackmail in particular, the technology for
analysing the security of even classical cryptographic constructions
against quantum attack is limited. There is extensive literature on
the security of constructions as varied as the Luby-Rackoff block
cipher, psuedorandom functions, public-key cryptosystems, and modern
block cipher modes. Many of these are provably secure against classical
attack assuming that some underlying primitive is secure. Little progress
has been made in defining security against quantum attacks. Classical
attacks can take many forms (e.g.\ known-plaintext attacks, adaptive
chosen-ciphertext attacks, etc.); this taxonomy will need to be extended
to meaningfully discuss security in a post-quantum world (e.g.\ what
is a non-adaptive chosen-quantum-ciphertext attack, and when is it
relevant?). Even less progress has been made in proving the security
of cryptographic constructions against quantum attack; the best we
can say is that no one has found generic attacks better than Grover's
algorithm.

Until the basic technology for analyzing the security of constructions
like OAEP against quantum attacks is in place, it will be difficult
to make rigorous statements about the security of even straightforward
quantum cryptographic constructions like quantum blackmail.

\section{Acknowledgments}

I wrote most of this paper at the Massachusetts Institute of Technology,
where I was supported by the Department of Defense (DoD) through the
National Defense Science \& Engineering Graduate Fellowship (NDSEG)
Program. An earlier version of this work appears in my thesis \cite[ chapter 7]{LutoThesis}.

\bibliographystyle{utphys}
\bibliography{/home/luto/Documents/Research/Thesis/v1.1/thesis}

\end{document}